\documentclass[final]{l4dc2026}

\title[Safe Planning with Robust Conformal
Prediction]{When Environments Shift: Safe Planning with
Generative Priors and Robust Conformal
Prediction}
\usepackage{times}
\usepackage{amsmath}    
\usepackage{amssymb}
\usepackage[nameinlink,noabbrev]{cleveref}
\usepackage{xcolor}
\usepackage[colorinlistoftodos]{todonotes}
\usepackage{algorithm}
\usepackage[noend]{algcompatible}
\usepackage{amsfonts}
\allowdisplaybreaks

\usepackage{booktabs}
\usepackage{pgfplots}
\usepackage{caption}
\usepackage{subcaption}
\pgfplotsset{compat=1.18}

\usepackage{ragged2e}

\author{%
 \Name{Kaizer Rahaman} 
 \Email{autoautokai@gmail.com}\\
 \addr  Department of Electrical Engineering, Indian Institute of Technology Kharagpur, India
 \AND
 \Name{Jyotirmoy V. Deshmukh} 
 \Email{jdeshmuk@usc.edu}\\
 \addr Thomas Lord Department of Computer Science, University of Southern California, USA
 \AND
 \Name{Ashish R. Hota} 
 \Email{ahota@ee.iitkgp.ac.in}\\
 \addr  Department of Electrical Engineering, Indian Institute of Technology Kharagpur, India
 \AND
 \Name{Lars Lindemann} \Email{llindemann@ethz.ch}\\
 \addr Automatic Control Laboratory, ETH Zürich, Switzerland
}

\begin{document}

\maketitle

\begin{abstract}%

Autonomous systems  operate in environments that may change over time. An example is the control of a self-driving vehicle among pedestrians and human-controlled vehicles whose behavior may change based on factors such as traffic density, road visibility, and social norms. Therefore, the environment encountered during deployment rarely mirrors the environment and data encountered during training -- a phenomenon known as distribution shift -- which can undermine the  safety of autonomous systems. Conformal prediction (CP) has recently been used along with  data from the training environment to provide prediction regions that capture the behavior of the environment with a desired probability. When embedded within a model predictive controller (MPC), one can provide probabilistic safety guarantees, but only when the deployment and training environments coincide. Once a distribution shift occurs, these guarantees collapse. We propose a planning framework that is robust under distribution shifts by: (i) assuming that the underlying data distribution of the environment is parameterized by a nuisance parameter, i.e., an observable, interpretable quantity such as traffic density, (ii) training a conditional diffusion model that captures distribution shifts as a function of the nuisance parameter, (iii) observing the nuisance parameter online and generating cheap, synthetic data from the diffusion model for the observed nuisance parameter, and (iv) designing an MPC that embeds CP regions constructed from such synthetic data. Importantly, we account for discrepancies between the underlying data distribution and the diffusion model by using robust CP.  Thus, the plans computed using robust CP enjoy probabilistic safety guarantees, in contrast with plans obtained from a single, static set of training data. We empirically demonstrate safety under diverse distribution shifts in the ORCA simulator.
\end{abstract}


\section{Introduction}
\label{sec:introduction}

Autonomous agents increasingly operate in dynamic, open-world environments where deployment-time data can differ from training data—a phenomenon known as \emph{distribution shift}. Such shifts arise from temporal variations, changing environmental conditions, or evolving agent behaviors, and can severely degrade the reliability of learning-based systems. In this paper, we  focus on distribution shifts in dynamic environments consisting of uncontrollable agents, where even mild mismatches between training and deployment conditions can compromise safety  in domains such as autonomous driving~\citep{filos2020can,sikar2024evaluation,arasteh2025validity} and robotics~\citep{paudel2022learning}.

Motion planning has been explored through both reactive and predictive paradigms~\citep{fox2002dynamic,mitsch2013provably,dimarogonas2006feedback,tanner2003nonholonomic}.  Reactive methods respond myopically to an autonomous agent's surrounding but typically lack optimality, whereas predictive methods anticipate future states of uncontrollable agents towards optimality. However, predictions of future states may be unreliable and result in unsafe behavior. In response, \emph{conformal prediction} (CP) was recently used to construct prediction regions that capture uncertainty in predicted future trajectories of uncontrollable agents  with a desired probability ~\citep{shafer2008tutorial,cleaveland2024conformal}. Building on this, \citep{lindemann2023safe} introduced \emph{Conformal MPC}, which integrates  CP regions within an MPC framework to achieve probabilistic safety. However, previously calibrated CP sets may become invalid under distributions shifts, resulting in  safety violations.

\refstepcounter{figure}
\label{fig:overview}

\begin{figure}[ht]
    \centering
    \includegraphics[width=\textwidth]{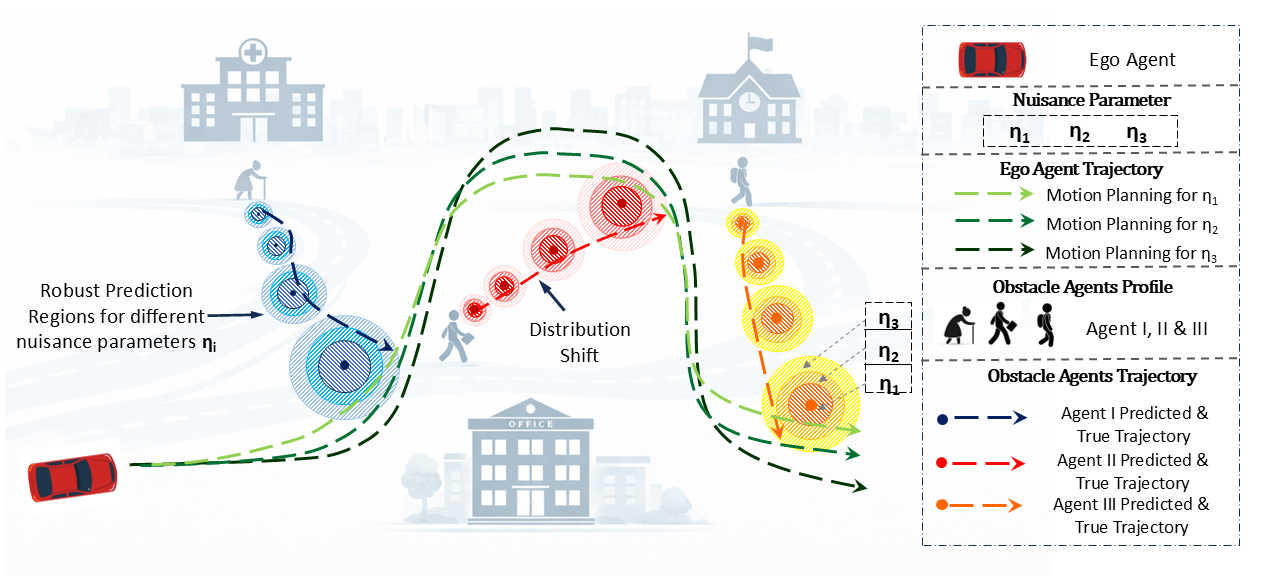}
    \par
    \justifying
    \small
    \noindent Figure \thefigure: The figure shows an ego agent interacting with three dynamic obstacles, illustrating true and predicted trajectories with prediction regions and the ego’s planned trajectories under three nuisance parameters. It highlights that only the trajectory corresponding to the underlying true nuisance parameter achieves collision avoidance, as it provides coverage for the real-time test environment.
\end{figure}

Instead, we propose a \emph{robust conformal prediction} framework for predictive motion planning that ensures safety under distribution shifts. Our underlying assumption is that the distribution of the environment is parameterized by a \emph{nuisance parameter}, i.e., an observable, interpretable quantity such as traffic density  or road visibility, which provides a structured representation of environmental variations. We then utilize a model of natural variation (MNV) ~\citep{robey2020model}  that generates inexpensive synthetic data approximating the deployment-time distribution. This enables the construction of CP regions that extend standard CP guarantees to shifted environments and facilitates safe motion planning, see Figure \ref{fig:overview}. Our contributions are summarized as follows: (1) We train a conditional diffusion model that learns the dependence of the nonconformity score -- the central element in the study of CP -- on the nuisance parameter.  (2) We design an MPC that uses robust CP regions that are constructed from cheap, synthetic data generated by the MNV. These robust CP regions account for the discrepancy between the underlying deployment distribution and the diffusion model and guarantee probabilistic safety. (3) We present empirical results in the ORCA simulator to show valid safety coverage under diverse distribution shifts. In comparison, we illustrate that the method from \cite{lindemann2023safe} results in more safety violations. We also demonstrate the need for robust CP even when we use synthetic calibration data that is close to the deployment data.

\textbf{Related Work.}  Motion planning in dynamic environments has been studied via sampling-based approaches~\citep{kalluraya2022multi,majd2021safe,aoude2013probabilistically,renganathan2023risk} and receding-horizon planning~\citep{wei2022moving,wang2022group,thomas2021probabilistic}. Model Predictive Control (MPC) provides a receding horizon framework that can incorporate robot dynamics, and can be used for collision avoidance, both in static environments~\citep{zhang2019trajectory, zhang2020optimization}, and in dynamic environments through robust and stochastic MPC formulations~\citep{dixit2021risk}. Yet, such approaches rely on accurate knowledge of  the dynamics and strategies of uncontrollable agents, and safety guarantees may degrade when their behavior changes over time. To mitigate this limitation, recent work has introduced distributionally robust motion planning, in which stochastic collision avoidance constraints are required to hold over an {ambiguity set} or a family of probability  distributions~\citep{hakobyan2019risk, navsalkar2023data, zolanvari2023wasserstein}. 

More related to this paper, prior work has embedded CP regions that capture the behavior of uncontrollable agents with a desired probability into the planner by either assuming that the environment does not change \citep{lindemann2023safe,tonkens2023scalable,chen2021reactive} or by using adaptive CP \citep{dixit2023adaptive,yao2024sonic,shin2025egocentric,sheng2024safe}, which provides average-time guarantees that do not ensure pointwise safety as we do by leveraging robust CP and cheap synthetic data generated by a MNV. CP was also used for reinforcement learning \citep{yao2024sonic,sun2023conformal}, control barrier functions \citep{zhang2025conformal,hsu2025statistical}, and sampling-based search \citep{sheng2024safe},  see \citep{lindemann2025formal} for a survey.

\section{Problem Formulation}
\label{sec:formulation}

\textbf{System Description.} We consider a discrete-time dynamical system of the form
\begin{equation}
    x_{t+1} = f(x_t, u_t), \qquad x_0 = \zeta,
    \label{eq:system-dynamics}
\end{equation}
where $x_t \in \mathcal{X} \subseteq \mathbb{R}^n$ denotes the (ego)  agent's state at time~$t$, $u_t \in \mathcal{U} \subseteq \mathbb{R}^m$ is the control input, and $f:\mathbb{R}^n\times\mathbb{R}^m\!\to\!\mathbb{R}^n$ represents the system dynamics. 
The ego  agent operates in a dynamic environment populated by $N$ uncontrollable agents with states $Y_t = (Y_{t,1},\dots,Y_{t,N}) \in \mathbb{R}^{Nn}$ whose  trajectories $(Y_0, Y_1, \dots)$ are unknown and evolve according to the stochastic process   $(Y_0, Y_1, \dots) \sim \mathcal{D}_{\mathrm{test}}$
where $\mathcal{D}_{\mathrm{test}}$ is  an (unknown) trajectory distribution encountered during online deployment.

\textbf{Training Distribution.} Let $\mathcal{D}_{\mathrm{train}}$ denote the nominal (offline) distribution from which data can be collected in the form of trajectories to train and calibrate (via conformal prediction) a predictor for the behavior of uncontrollable agents. 
We refer to a \emph{distribution shift} when $\mathcal{D}_{\mathrm{test}} \neq \mathcal{D}_{\mathrm{train}}$, that is, when trajectory data encountered during deployment deviates from data collected for training. We measure such distribution shifts via a statistical distance $\mathsf{d}(\mathcal{D}_{\mathrm{test}}, \mathcal{D}_{\mathrm{train}})$, where $\mathsf{d}$ may denote the Wasserstein distance, KL divergence, or total variation metric. In our setting, distribution shifts arise possibly due to changes in traffic density, road visibility, social norms and couplings between agents.  Consequently, safety guarantees for $\mathcal{D}_{\mathrm{train}}$ may no longer hold under $\mathcal{D}_{\mathrm{test}}$.

\textbf{Idealized Safe Planning Problem.} Let $c:\mathbb{R}^n\times\mathbb{R}^{Nn}\!\to\!\mathbb{R}$ be a Lipschitz continuous constraint function that indicates safety, e.g., the collision avoidance constraint $c(x_t,Y_t)=\min_j\|x_t-Y_{t,j}\|-\epsilon\ge 0$ imposes that the distance between the states (e.g., positions) of the ago agent and any other (obstacle) agent is at least some safety margin $\epsilon$. We aim to optimize performance while maintaining safety with high probability, leading to a chance-constrained optimization problem:
\begin{subequations}\label{eq:ideal-opt}
\begin{align}
    \min_{u_0,\dots,u_{T-1}} \;  J(x,u) \quad 
    \text{s.t. } \;
    & x_{t+1} = f(x_t,u_t), \quad t=0,\dots,T-1, \\
    & \mathbb{P}_{\mathcal{D}_{\mathrm{test}}}
      \!\big( c(x_t,Y_t)\ge 0,\;\forall t\in\{0,\dots,T\}\big)
      \ge 1-\delta, \label{eq:chance}
\end{align}
\end{subequations}
where $J(x,u)$ denotes a user-defined cost (e.g., tracking error or energy consumption), $\delta\in(0,1)$ is a risk tolerance parameter, and $\mathcal{D}_{\mathrm{test}}$ represents the actual (possibly shifted) distribution of the environment during deployment. Solving~\eqref{eq:ideal-opt} exactly is generally intractable since $\mathcal{D}_{\mathrm{test}}$ is unknown and may differ from the nominal training distribution $\mathcal{D}_{\mathrm{train}}$. Furthermore, the stochastic process $\mathcal{D}_{\mathrm{test}}$ is high-dimensional and nonlinear, making the chance constraint computationally prohibitive. 

We adopt a \emph{data-driven} uncertainty quantification approach using \emph{conformal prediction} (CP). Conceptually, we follow \cite{lindemann2023safe} in  that we use trajectory predictors for uncontrollable agents $Y$, construct  prediction regions with CP using calibration data from $\mathcal{D}_{\mathrm{train}}$, and then use these to design a model predictive controller for \eqref{eq:system-dynamics}.  In the absence of distribution shifts, guarantees of the form $\mathbb{P}_{\mathcal{D}_{\mathrm{test}}}^{K+1}
      \!\big( c(x_t,Y_t)\ge 0,\;\forall t\in\{0,\dots,T\}\big)
      \ge 1-\delta$ can be obtained where $K$ is the number of calibration data.\footnote{Note that the $\mathbb{P}_{\mathcal{D}_{\mathrm{test}}}^{K+1}$ is the $(K+1)$-fold product probability measure of $\mathbb{P}_{\mathcal{D}_{\mathrm{test}}}$ -- their relation and  the accuracy to which the chance constraint in \eqref{eq:chance} is approximated is discussed in more detail \cite{lindemann2025formal}.}  Due to distribution shifts, however, we cannot simply use calibration data from $\mathcal{D}_{\mathrm{train}}$, motivating the need for robust methods that account for  distribution shifts.

\section{Preliminaries}
\label{sec:prelim}

\textbf{Vanilla Conformal Prediction.}   Let $R^{(0)},\ldots,R^{(K)} \sim \mathcal{R}_\mathrm{test}$ be $K+1$ i.i.d. random variables, referred to as the non-conformity score.  For instance, in a supervised learning setting with inputs $X^{(k)}$ and outputs $Z^{(k)}$, a natural choice is $R^{(k)} := \| Z^{(k)} - \mu(X^{(k)}) \|$ where $\mu$ denotes the predictor; larger values of the nonconformity score indicate poorer predictive performance. The goal of CP is to compute a bound $C$ from the calibration data $R^{(1)},\ldots,R^{(K)}$ such that the test data $R^{(0)}$ satisfies $\mathbb{P}_{\mathcal{R}_\mathrm{test}}^{K+1}(R^{(0)} \le C) \ge 1-\delta$, where $1-\delta$ is a user-specified confidence level. Specifically, $C := \text{Quantile}_{(1+1/K)(1-\delta)}(R^{(1)},\ldots,R^{(K)})$, i.e., the corrected $(1-\delta)$-quantile of the empirical distribution over the calibration data, achieves the above probabilistic guarantee \citep{vovk2005algorithmic}.

\textbf{Robust Conformal Prediction.}  When a distribution shift occurs, standard conformal prediction (CP) coverage guarantees may no longer hold. Let $R^{(1)}, \dots, R^{(K)} \sim \mathcal{R}_{\mathrm{train}}$ and $R^{(0)} \sim \mathcal{R}_{\mathrm{test}}$ be independent, where the nonconformity scores are induced by $\mathcal{D}_{\mathrm{train}}$ and $\mathcal{D}_{\mathrm{test}}$, respectively. Robust Conformal Prediction (RCP) seeks to ensure valid coverage for any $\mathcal{R}_{\mathrm{test}}$ within an ambiguity set $
\mathcal{U}_{r} = \{\; \mathcal{Q} : \mathsf{d}(\mathcal{Q}, \mathcal{R}_{\mathrm{train}}) \le r \;\}$
where $r>0$ bounds the admissible shift magnitude~\citep{cauchois2024robust,aolaritei2025conformal}. The robust coverage constraint can be written as
\begin{equation}
    \inf_{\mathcal{R}_{\mathrm{test}} \in \mathcal{U}_{r}}
    \mathbb{P}_{\mathcal{R}_{\mathrm{test}},\mathcal{R}_{\mathrm{train}}}^{K+1}
    \big( R^{(0)} \le C^{\mathrm{rob}} \big)
    \ge 1 - \delta,
    \label{eq:robust_cp_constraint}
\end{equation}
where $C^{\mathrm{rob}}$ denotes a \emph{robust quantile}, often expressed as an inflation of the empirical quantile $C$  as:
\begin{equation}
    C^{\mathrm{rob}} = C + \Delta_\mathsf{d}(r;  R^{(1)}, \dots, R^{(K)}).
    \label{eq:robust_additive_form}
\end{equation}
Here, $\Delta_\mathsf{d}(r; R^{(1)}, \dots, R^{(K)})$ represents a robustness correction that depends on the ambiguity radius $r$ and, in general, on the nonconformity scores from the calibration data $R^{(1)}, \dots, R^{(K)}$. The form of $\Delta_\mathsf{d}(\cdot)$ depends on how the distributional uncertainty is modeled: (i) in the $f$-divergence formulation \citep{cauchois2024robust}, $\Delta_\mathsf{d}$ is implicitly determined through a convex optimization problem over distributions constrained by an $f$-divergence ambiguity set; and (ii) in the L\'evy--Prokhorov (LP) formulation \citep{aolaritei2025conformal}, $\Delta_\mathsf{d}$ arises from coupled perturbation radii $(\varepsilon,\rho)$ associated with local (infimum-Wasserstein) and global (total-variation) shifts. For clarity of exposition, we use the simplified additive form~\eqref{eq:robust_additive_form} to represent the effective robustness margin, while acknowledging that, in general, $\Delta_\mathsf{d}$ may depend on multiple factors such as the divergence type, coupled radii (LP metric) and  calibration samples. We refer the reader to Appendix~\ref{app:robust_cp_details} for a detailed treatment.

\textbf{Conditional Diffusion Models.} We use conditional diffusion models  for modeling conditional data distributions ~\citep{ho2020denoising, han2022card}. Let $s_0$ be a variable that we aim to describe, later corresponding to trajectories $Y$ and nonconformity scores $R$. The variable $s_0$  may depend on a contextual parameters $\mathbf{c}$, later corresponding to our nuisance parameter. Conditional diffusion models  are distributions \(p_\theta(s_0 \mid \mathbf{c})\) which are parameterized by a variable $\theta$ that will be learned.  Following the classification and regression diffusion (CARD) models \citep{han2022card}, the contextual variable 
\(\mathbf{c}\) is passed through a pre-trained conditional mean encoder \(f_{\phi}(\mathbf{c})\) which estimates $\mathbb{E}[s_0\mid \mathbf{c}]$ to anchor the diffusion process around a deterministic conditional mean, allowing the generative model to focus on modeling residual uncertainty. These models construct a generative process by progressively corrupting data  with Gaussian noise and learning to reverse this process. Full expressions for this forward (noising) process \(q(s_{j} \mid s_0,  \mathbf{c}\)) and variance schedule \(\beta_{j}\) are provided in Appendix~\ref{app:dif_model}, where $j$ denotes the $j$-th diffusion step. The reverse (denoising) process parametrizes a conditional Gaussian transition:
\begin{equation}
\small{    p_\theta(s_{j-1} \mid s_{j}, \mathbf{c})
    =
    \mathcal{N}\!\left(
        \mu_\theta(s_{j}, j, f_{\phi}(\mathbf{c})),
        \Sigma_{j}
    \right),
    \;
    \mu_\theta(s_{j}, j, f_{\phi}(\mathbf{c}))
    =
    \frac{1}{\sqrt{\alpha_{j}}}
    \big(
        s_{j}
        - \frac{\beta_{j}}{\sqrt{1-\bar{\alpha}_{j}}}\,
          \epsilon_\theta(s_{j}, j, f_{\phi}(\mathbf{c}))
    \big),}
    \label{eq:p_theta_cond}
\end{equation}
where $\epsilon_\theta(s_j, j, f_{\phi}(\mathbf{c}))$ predicts the additive noise introduced at diffusion step $j$. Here, $\alpha_j = 1 - \beta_j$ represents the per-step noise retention factor and $\bar{\alpha}_j = \prod_{k=1}^{j} \alpha_k$ is its cumulative product controlling the effective signal-to-noise ratio across steps, see \citep{ho2020denoising} for details. This model thus learns a smooth, context-dependent score field over the contextual parameter \(\mathbf{c}\) which is able to approximate underlying, unknown distributions. More details are provided in Appendix~\ref{app:dif_model}.

\section{Safe Planning with Generative Priors and Robust Conformal Prediction}
\label{sec:RCP_TP}

Our central premise is that  the uncontrollable trajectories are not governed by a single, fixed  distribution $\mathcal{D}_{\mathrm{test}}$, but by a family of distributions indexed by underlying latent factors $\zeta$ -- termed \emph{natural variations} -- that capture phenomena such as environmental changes or shifts in agent behavior that cannot be directly observed but systematically alter the data distribution. These  latent factors parameterize a conditional data distribution $\mathcal{D}_{\mathrm{test}}(\zeta)$ that may be encountered during deployment. In practice, however, the latent factor $\zeta$ may not be directly observable. We introduce an \emph{observable nuisance parameter} $\eta$ which is a function of the latent factor $\zeta$, i.e., the nuisance parameter is governed by  $\eta(\zeta)$ where we omit dependency on $\zeta$ when appropriate. We use $\eta$ to learn a distribution that serves as a proxy for these latent factors  and  data-generating distribution. Indeed,  we present two variations in Sections \ref{sec:4_1} and \ref{sec:4_2} where we learn trajectory and nonconformity score distributions $\mathcal{D}_{\mathrm{train}}(\eta)$ and $\mathcal{R}_{\mathrm{train}}(\eta)$ to approximate $\mathcal{D}_{\mathrm{test}}(\zeta)$ and $\mathcal{R}_{\mathrm{test}}(\zeta)$ so that, ideally, we obtain $\mathcal{D}_{\mathrm{test}}(\zeta) {=}\mathcal{D}_{\mathrm{train}}(\eta(\zeta))$ and $\mathcal{R}_{\mathrm{test}}(\zeta) {=}\mathcal{R}_{\mathrm{train}}(\eta(\zeta))$. The nuisance parameter acts as a contextual descriptor of the current operating regime, and it may correspond to traffic density, road visibility, and others. Our approach will use this nuisance parameter $\eta$ to deal with distribution shifts and follow a three step procedure. At each time, we predict the environment, construct robust prediction regions from  calibration data sampled from  $\mathcal{D}_{\mathrm{train}}(\eta)$ or $\mathcal{R}_{\mathrm{train}}(\eta)$, and solve an MPC iteratively.

 \subsection{Robust Prediction Regions for Uncontrollable Agents with Generative Trajectory Priors} \label{sec:4_1} We now present the construction of prediction regions that contain  uncontrollable agents with a  probability of at least $1-\delta$. We use robust conformal prediction  with calibration data  from the training distribution $\mathcal{D}_{\mathrm{train}}(\eta)$. The distribution $\mathcal{D}_{\mathrm{train}}(\eta)$ here corresponds to a conditional diffusion model \(p_\theta(Y \mid \mathbf{c})\) with observable conditioning variables $\mathbf{c}:=\eta$. The model \(p_\theta(Y \mid \mathbf{c})\) can be trained with data from $\mathcal{D}_{\mathrm{test}}(\zeta)$ for a single or a range of values of $\zeta$. Ideally though, data used for training should have a diverse mix of data from different latent factors $\zeta$.  Alternatively, the data used to train \(p_\theta(Y \mid \mathbf{c})\) could come from an open access dataset \citep{o2024open}.

\textbf{Trajectory Prediction. } At each timestep $t$, we use a learned trajectory predictor to produce $\tau$-step-ahead predictions $\hat{Y}_{\tau \mid t}$ of the environment state $Y_\tau$ for all prediction times $\tau \in \{t+1, \dots, t+H\}$, where $H$ denotes the MPC prediction horizon. The predictor can be any type of predictor and could be trained on the same data that we have trained the training distribution $\mathcal{D}_{\mathrm{train}}(\eta)$ on. As such, we can even train a conditional predictor that provides predictions $\hat{Y}_{\tau \mid t}(\eta)$ conditioned on $\eta$. Our goal now is to compute constants $C_{\tau\mid t,i}^{\mathrm{rob}}\ge 0$ that define prediction regions of the form
\begin{equation}
\mathbb{P}_{\mathcal{R}_{\mathrm{test}},\mathcal{R}_{\mathrm{train}}}^{K+1}
    \big( \|\hat{Y}_{\tau \mid t,i}- Y_{\tau,i}||\le C_{\tau\mid t,i}^{\mathrm{rob}}, \; \forall (t,\tau,i)\in \mathcal{S}  \big)
    \ge 1 - \delta
    \label{eq:robust_CP_regions}
\end{equation}
where $\mathcal{S}=\{1,\hdots,T-1\}\times \{t+1,\hdots,t+H\} \times \{1,\hdots,N\}$. We note that the prediction region in \eqref{eq:robust_CP_regions} provides simultaneous coverage over all time steps $t$, predictions $\tau$, and agents $i$. Nonconformity scores for simultaneous coverage were proposed in \citep{cleaveland2024conformal,sun2022copula}, but fail to provide coverage under distribution shift. The authors in \citep{zhao2024robust}  provided simultaneous coverage under distribution shifts captured by $f$-divergence measures. Here, we follow a similar idea, but  permit distribution shifts captured by the L\'evy--Prokhorov metric while using synthetic data from a conditional generative model to minimize conservatism which is typically induced by robust CP methods -- these advantages  will later be illustrated in our experiments.

\textbf{Nonconformity score.} We can obtain the guarantees in equation \eqref{eq:robust_CP_regions}  by the nonconformity score 
\begin{align}\label{eq:nonconformity}
    R^{(k)}=\max_{(t,\tau,i)\in\mathcal{S}} \frac{\|\hat{Y}_{\tau \mid t,i}^{(k)}- Y_{\tau,i}^{(k)}||}{\sigma_{\tau \mid t,i}}
\end{align}
where $Y^{(1)}, \hdots,Y^{(K)}\sim \mathcal{D}_{\mathrm{train}}(\eta)$ are $K$ training trajectories, $\hat{Y}_{\tau \mid t,i}^{(1)},\hdots, \hat{Y}_{\tau \mid t,i}^{(K)}$ are the corresponding predictions, and $\sigma_{\tau \mid t,i}>0$ are normalization constants, see e.g., \citep{cleaveland2024conformal,yu2026signal} for their computation. The choice of $\sigma_{\tau \mid t,i}$ does not affect validity of \eqref{eq:robust_CP_regions}, but well chosen $\sigma_{\tau \mid t,i}$ result in small prediction regions.  This nonconformity score quantifies the discrepancy between the predicted trajectory $\hat{Y}_{\tau \mid t,i}^{(k)}$ and the true trajectory $Y_{\tau,i}^{(k)}$ for each time step $t$, prediction $\tau$, and agent $i$. If there is no distribution shift,  the choice  $C_{\tau\mid t,i}^{\mathrm{rob}}:=C\sigma_{\tau \mid t,i}$ ensures that \eqref{eq:robust_CP_regions} holds.

\textbf{Distribution shift for the Nonconformity score.} 
 The nonconformity scores $\{R^{(i)}\}_{i=1}^K$ follow a nonconformity score distribution $\mathcal{R}_{\mathrm{train}}$, i.e., $R^{(1)},\hdots, R^{(K)}\sim \mathcal{R}_{\mathrm{train}}$, which is induced by the  distribution $\mathcal{D}_{\mathrm{train}}$.\footnote{We drop the dependency on the latent factor $\zeta$ and the nuisance parameter $\eta$ when clear from the context.} However, we need to consider the  nonconformity score distribution $\mathcal{R}_{\mathrm{test}}$ for the nonconformity score $R^{(0)}$ during deployment for which $\mathcal{R}_{\mathrm{train}} \neq \mathcal{R}_{\mathrm{test}}$ may hold. To apply robust conformal prediction as introduced in Section \ref{sec:prelim}, we need to know  the statistical distance $\mathsf{d}\big(\mathcal{R}_{\mathrm{test}}, \mathcal{R}_{\mathrm{train}}\big)$, which we can estimate in practice or derive analytical bounds in the case of diffusion models (details provided later). The next result follows directly by applying robust conformal prediction and is a generalization of \citep[Lemma 3]{zhao2024robust} by allowing more general statistical distances, such as the Wasserstein and L\'evy--Prokhorov metric from~\citep{aolaritei2025conformal}.
 \begin{lemma}\label{cor1}
    Let $\delta\in(0,1)$ be a failure probability and $Y^{(0)} \sim \mathcal{D}_{\mathrm{test}}(\zeta)$ and $Y^{(1)},\hdots, Y^{(K)} \sim \mathcal{D}_{\mathrm{train}}(\eta)$ be test and training trajectories for the latent factor $\zeta$  and nuisance parameter  $\eta$. Let $R^{(0)}\sim\mathcal{R}_{\mathrm{test}}$ and $R^{(1)},\hdots, R^{(K)}\sim \mathcal{R}_{\mathrm{train}}$ follow \eqref{eq:nonconformity} where $\mathcal{R}_{\mathrm{test}}$ and $ \mathcal{R}_{\mathrm{train}}$ are the distributions induced by $\mathcal{D}_{\mathrm{test}}$ and $\mathcal{D}_{\mathrm{train}}$.  If $r\ge 0$ is such that $\mathsf{d}\big(\mathcal{R}_{\mathrm{test}}, \mathcal{R}_{\mathrm{train}}\big)\le r$, then the choice of $C_{\tau\mid t,i}^{\mathrm{rob}}:=\sigma_{\tau \mid t,i}\cdot(C+\Delta_\mathsf{d}(r;R^{(1)}, \dots, R^{(K)}))$ guarantees that the prediction region in \eqref{eq:robust_CP_regions} holds. 
    \begin{proof}
        Since $\mathsf{d}\big(\mathcal{R}_{\mathrm{test}}, \mathcal{R}_{\mathrm{train}}\big)\le r$ holds, robust conformal prediction  guarantees via \eqref{eq:robust_cp_constraint} that 
        \begin{equation}\label{eq:proof}
\mathbb{P}_{\mathcal{R}_{\mathrm{test}},\mathcal{R}_{\mathrm{train}}}^{K+1}
    \big( R^{(0)} \le C+\Delta_\mathsf{d}(r;R^{(1)}, \dots, R^{(K)}) \big)
    \ge 1 - \delta.
\end{equation}
Plugging $R^{(0)}$ from equation \eqref{eq:nonconformity} into \eqref{eq:proof} gives us that
        \begin{align}
    &\mathbb{P}_{\mathcal{R}_{\mathrm{test}},\mathcal{R}_{\mathrm{train}}}^{K+1}
    \Big( \max_{(t,\tau,i)\in\mathcal{S}} \frac{\|\hat{Y}_{\tau \mid t,i}^{(0)}- Y_{\tau,i}^{(0)}||}{\sigma_{\tau \mid t,i}} \le C+\Delta_\mathsf{d}(r;R^{(1)}, \dots, R^{(K)}) \Big)
    \ge 1 - \delta\\
    \Leftrightarrow\;\;\; &\mathbb{P}_{\mathcal{R}_{\mathrm{test}},\mathcal{R}_{\mathrm{train}}}^{K+1}
    \Big( \|\hat{Y}_{\tau \mid t,i}^{(0)}- Y_{\tau,i}^{(0)}|| \le \sigma_{\tau \mid t,i}(C+\Delta_\mathsf{d}(r;R^{(1)}, \dots, R^{(K)})), \; \forall (t,\tau,i)\in \mathcal{S}  \Big)
    \ge 1 - \delta
\end{align}
from which equation \eqref{eq:robust_CP_regions} follows by the choice of $C_{\tau\mid t,i}^{\mathrm{rob}}:=\sigma_{\tau \mid t,i}\cdot(C+\Delta_\mathsf{d}(r;R^{(1)}, \dots, R^{(K)}))$.
    \end{proof}

\end{lemma}

 It is easy to see that excessively large bounds $\Delta_\mathsf{d}(r;R^{(1)}, \dots, R^{(K)})$ in Lemma \ref{cor1} can lead to overly conservative prediction regions, which in turn can make the resulting MPC controller (presented in a later section) conservative. For instance, this may be the case when only a fixed distribution $\mathcal{D}_{\mathrm{train}}$ is used, e.g., as in \citep{zhao2024robust}. Effectively, this case corresponds to a single nuisance parameter and motivates our approach where we can generate cheap, synthetic data from a conditional diffusion model $\mathcal{D}_{\mathrm{train}}(\eta)$ that can reduce the bound $\Delta_\mathsf{d}(r;R^{(1)}, \dots, R^{(K)})$.

\subsection{Efficient Prediction Regions with Generative Nonconformity Score Priors}
\label{sec:4_2} Generating trajectories with diffusion models is computationally expensive and time consuming, potentially prohibitive for real-time application. Instead of generating high-dimensional trajectory data via a learned distribution $\mathcal{D}_{\mathrm{train}}(\eta)$, we here propose to generate nonconformity scores via learning the distribution $\mathcal{R}_{\mathrm{train}}(\eta)$ directly.  Indeed, we let $\mathcal{R}_{\mathrm{train}}(\eta)$ be based on a conditional diffusion model $p_\theta(\bar{R} \mid \mathbf{c})$ where now the conditioning variable is $\mathbf{c} = [\, \eta,\, t,\, \tau ,\, i \,]$. This is so that  $p_\theta(\bar{R} \mid \mathbf{c})$ models the prediction error $\|\hat{Y}_{\tau \mid t,i}^{(0)}- Y_{\tau,i}^{(0)}||$ for trajectories  $Y^{(0)}\sim \mathcal{D}_{\mathrm{test}}(\zeta)$.  We note that we condition not only on the nuisance parameter $\eta$, but also on the current time $t$, predictions $\tau$, and agents $i$. This way, we can approximate the nonconformity score in  \eqref{eq:nonconformity} as
\begin{align}\label{eq:nonconformity_42}
    R^{(k)}=\max_{(t,\tau,i)\in\mathcal{S}} \frac{\bar{R}_{\tau|t,i}^{(k)}}{\sigma_{\tau \mid t,i}}
\end{align}
where $\bar{R}_{\tau|t,i}^{(k)}$ is generated by $p_\theta(\bar{R} \mid \mathbf{c})$, i.e., $\bar{R}_{\tau|t,i}^{(k)}\sim p_\theta(\bar{R} \mid \mathbf{c})$. Training the diffusion model $p_\theta(\bar{R} \mid \mathbf{c})$ follows again the process described in Section \ref{sec:prelim}, with details presented in Appendix~\ref{app:dif_model}. Data used for training should  again have a diverse mix of data with different conditioning variables $\mathbf{c}$ to learn a continuous family of conditional distributions that generalize to unseen nuisance parameters, as we demonstrate in our experiments. The next lemma summarizes our results when the calibration nonconformity scores $R^{(1)},\hdots, R^{(K)}$ follow equation \eqref{eq:nonconformity_42} and not equation \eqref{eq:nonconformity}, while the test nonconformity score $R^{(0)}$ again follows \eqref{eq:nonconformity}. Since the test nonconformity score $R^{(0)}$ still follows \eqref{eq:nonconformity}, the proof of the lemma follows almost exactly the  same steps and is omitted.

 \begin{lemma}\label{cor2}
    Let $\delta\in(0,1)$ be a failure probability and $Y^{(0)} \sim \mathcal{D}_{\mathrm{test}}(\zeta)$ be the test trajectory for the latent factor $\zeta$. Let $R^{(0)}\sim\mathcal{R}_{\mathrm{test}}$ follow \eqref{eq:nonconformity} where $\mathcal{R}_{\mathrm{test}}$ is the distribution induced by $\mathcal{D}_{\mathrm{test}}$  and $R^{(1)},\hdots, R^{(K)}\sim \mathcal{R}_{\mathrm{train}}$ follow \eqref{eq:nonconformity_42} where $\mathcal{R}_{\mathrm{train}}$ is the distribution induced by $p_\theta(R \mid \mathbf{c})$ with $\mathbf{c} = [\, \eta,\, t,\, \tau ,\, i \,]$. If $r\ge 0$ is such that $\mathsf{d}\big(\mathcal{R}_{\mathrm{test}}, \mathcal{R}_{\mathrm{train}}\big)\le r$, then the choice of $C_{\tau\mid t,i}^{\mathrm{rob}}:=\sigma_{\tau \mid t,i}\cdot(C+\Delta_\mathsf{d}(r;R^{(1)}, \dots, R^{(K)}))$ guarantees that the prediction region in \eqref{eq:robust_CP_regions} holds. 
\end{lemma}

    The nonconformity score in \eqref{eq:nonconformity_42}, used for calibration, does not depend on  the trajectory predictor. The nonconformity score \eqref{eq:nonconformity}, used during deployment, on the other hand does and thereby defines prediction regions with center $\hat{Y}_{\tau \mid t,i}$ and size $C_{\tau\mid t,i}^{\mathrm{rob}}$, enabling control design in Section \ref{sec:MPC}.

 \subsection{Estimating Distribution Shifts for Generative Priors}
 
 Without any prior knowledge, it is not possible to know the distribution shift $r$ which is ultimately required to apply Lemmas \ref{cor1} and \ref{cor2}. Note also that $r$ may be different for different contexts $\mathbf{c}$. Generally, $r$ could be treated as a tuning knob that can be empirically adjusted to add robustness margins to the algorithms, e.g., as in robust control \citep{zhou1998essentials}. We instead aim to estimate the distribution shift  $r$ here.  For this purpose, we focus on estimating $r$ for Lemma \ref{cor2}, which we also implement in our experiments, while similar methods apply to estimating $r$ for Lemma \ref{cor1}. 
 
\textbf{Data under context $\mathbf{c}$ available.} We will now propose two methods to estimate $r$ when some test data  available. In the first method, we assume to have a few trajectories from $\mathcal{D}_{\mathrm{test}}(\zeta)$ -- and consequently $\mathcal{R}_{\mathrm{test}}(\zeta)$ -- available so that we can estimate $r$ directly, e.g., using \citep[Algorithm 1]{aolaritei2025conformal} for the LP metric or \citep{rubenstein2019practical} for the $f$-divergence. 

In the second method, we  obtain an upper bound for $r$ by quantifying the discrepancy between the true conditional distribution of the prediction error $\|\hat{Y}_{\tau \mid t,i}^{(k)}- Y_{\tau,i}^{(k)}||$ for trajectories  $Y^{(k)}\sim \mathcal{D}_{\mathrm{test}}(\zeta)$, which we denote by $\bar{\mathcal{R}}_{\mathrm{test}}(\zeta)$, and that generated by our diffusion model $p_\theta(\bar{R} \mid \mathbf{c})$. We employ the analytical 2-Wasserstein bound derived in~\citep[Theorem 1]{li2025non}, which states that
\begin{equation}\label{eq:analytic:bound}
    W_2\!\big(\bar{\mathcal{R}}_{\mathrm{test}}(\zeta),\,p_\theta(\bar{R}\mid \mathbf{c})\big)
    \;\le\;
    \sum_{j=1}^{T_{\sf diff}} \beta_j\,M(j)\,\sqrt{H(j)},
    \;
    M(j)\!=\!\exp\!\Big(\sum_{s\le j}\!(L_1(s)\!+\!L_2(s)\beta_s)\Big)
\end{equation}
where \(H(j)\) represents the normalized mean square error of the model's predicted diffusion noise, which we estimate on a small held-out validation set $\bar{R}_{\mathrm{test}}(\zeta)$ drawn from $\bar{\mathcal{R}}_{\mathrm{test}}(\zeta)$, $L_1,L_2$ are drift regularity constants of the forward process, and $T_{\sf diff}$ is the total number of diffusion steps. We refer the reader to Appendix~\ref{app:analytical_bound} for  details.   As we obtain an upper bound for $\|\hat{Y}_{\tau \mid t,i}^{(k)}- Y_{\tau,i}^{(k)}||$ this way, the choice of $C_{\tau\mid t,i}^{\mathrm{rob}}:=\sigma_{\tau\mid t,i}\Big(C+\Delta_{\mathsf d}\!\left(\tfrac{\varepsilon_W}{\sigma_{\min}};\,R^{(1)},\dots,R^{(K)}\right)\Big)$ with $\sigma_{\min}:=\min_{(t,\tau,i)\in\mathcal{S}}\sigma_{\tau\mid t,i}$ and $\varepsilon_W:=\sum_{j=1}^{T_{\sf diff}} \beta_j\,M(j)\,\sqrt{H(j)}$ now makes Lemma \ref{cor2} valid, see Appendix \ref{app:analytical_bound_CP} for a proof. 

\textbf{No data under context $\mathbf{c}$ available.} We may encounter a new context $\mathbf{c}^\star=(\eta^\star,\ldots)$ where we have no test data available under the nuisance parameter $\eta^\star$. Assume that we have datasets  $R_{\mathrm{test}}(\mathbf{c})$ or $\bar{R}_{\mathrm{test}}(\mathbf{c})$ drawn from $\mathcal{R}_{\mathrm{test}}(\zeta)$ or $\bar{\mathcal{R}}_{\mathrm{test}}(\zeta)$ available, respectively,  corresponding to the first or second method presented previously. To handle an ``unseen'' contexts $\mathbf{c}^\star$, we construct a context-specific test dataset $R_{\mathrm{test}}(\mathbf{c}^\star)$ (or $\bar{R}_{\mathrm{test}}(\mathbf{c}^\star)$) by interpolating between nearby seen contexts. We do so by aggregating data from the $I$ closest seen contexts $\{\mathbf c^{(1)},\ldots,\mathbf c^{(I)}\}$ to $\mathbf c^\star$, where closeness is defined with respect to a chosen metric on the nuisance space.  We then construct
${R}_{\mathrm{test}}(\mathbf{c}^\star) = \cup_{i=1}^I {R}_{\mathrm{test}}(\mathbf{c}^{(i)}) $ (or $\bar{R}_{\mathrm{test}}(\mathbf{c}^\star) = \cup_{i=1}^I \bar{R}_{\mathrm{test}}(\mathbf{c}^{(i)}) $) which serves as the validation dataset for estimating $r$ at the unseen context $\mathbf{c}^\star$ with one of the two previous methods discussed.

\subsection{Model Predictive Control with Robust Conformal Prediction Regions} \label{sec:MPC} Using these prediction sets, we iteratively solve the following MPC formulation:
\begin{subequations}\label{eq:mpc}
\begin{align}
    \min_{u_{t:t+H-1}} \;& J(x,u) \\
    \text{s.t. } & x_{\tau+1} = f(x_\tau,u_\tau), \quad \tau=t,\dots,t+H-1, \label{eq:mpc_dyn}\\
    & \inf_{ (Y_0,\hdots,Y_t,\hdots,Y_{t+H}) \in \mathcal{C}_{t}}c(x_\tau, Y_{\tau}) \ge 0, \quad   \tau=t+1,\dots,t+H, \label{eq:mpc_safe}\\
    & x_\tau \in \mathcal{X},\; u_\tau \in \mathcal{U}, \quad \tau=t,\dots,t+H-1,
\end{align}
\end{subequations}
where $\mathcal{C}_{t}:=\{Y| \|\hat{Y}_{\tau \mid t,i}- Y_{\tau,i}||\le C_{\tau\mid t,i}^{\mathrm{rob}}, \; \forall (\tau,i)\in \{t+1,\hdots,t+H\}\times\{1,\hdots,N\}\}$.

 \begin{theorem}\label{Theorem_1}
    (Closed-loop control.) Given the system in (\ref{eq:system-dynamics}), a failure probability $\delta\in(0,1)$, an unknown latent factor $\zeta$, and an observable nuisance parameter $\eta(\zeta)$. Assume that the prediction region in \eqref{eq:robust_CP_regions} is valid (e.g., via Lemmas \ref{cor1} or \ref{cor2} by use of $\eta$) and that the optimization problem \eqref{eq:mpc} is feasible at each time $t\in \{0,\hdots,T-1\}$, then the closed-loop system satisfies
    \begin{align}\label{eq:guarantee}
\mathbb{P}_{\mathcal{R}_{\mathrm{test}},\mathcal{R}_{\mathrm{train}}}^{K+1}\big(c(x_t,Y_t)\ge 0, \forall t\in\{1,\hdots,T\}\big)\ge 1-\delta.
    \end{align}
    where $Y \sim \mathcal{D}_{\mathrm{test}}(\zeta)$ is the test trajectory for the latent factor $\zeta$. 
\end{theorem}

\begin{proof}
The guarantee in \eqref{eq:guarantee} directly follows since the optimization problem \eqref{eq:mpc} is feasible at each time $t\in \{0,\hdots,T-1\}$ and since the constraint \eqref{eq:mpc_safe} uses the set $\mathcal{C}_t$ that is such that  $\mathbb{P}_{\mathcal{R}_{\mathrm{test}},\mathcal{R}_{\mathrm{train}}}^{K+1}\big( Y\in \mathcal{C}_t, \forall t\in\{1,\hdots,T\} )\ge 1-\delta$, which is guaranteed since the prediction regions in \eqref{eq:robust_CP_regions} are valid.
\end{proof}

\section{Case Study}
\label{sec:Case_Study}

 All simulations were executed on a 64-bit \texttt{x86\_64} workstation equipped with an Intel(R) Core(TM) i7-10700 CPU (2.90~GHz) and 32~GB RAM. The CasADi optimization framework from~\citep{andersson2019casadi} was employed in an Anaconda-Python environment. In the reminder, we compute efficient prediction regions with the method in Section \ref{sec:4_2} and  test and compare four instantiation of the MPC: \textbf{Case-0:} This baseline is from \cite{lindemann2023safe} which uses the nominal quantile $C$ without robustification computed using calibration data from a single distribution $\mathcal{D}_{\mathrm{train}}$ by fixing a single nuisance parameter. \textbf{Case-1:} This baseline uses the quantile $C$ without robustification, but  computed using synthetic data from $p_\theta$ under the observed $\eta$. \textbf{Case-2:} This baseline uses the robust quantile $C_{\tau\mid t,i}^{\mathrm{rob}}:=\sigma_{\tau \mid t,i}\cdot(C+\Delta_\mathsf{d}(r;R^{(1)}, \dots, R^{(K)}))$  where $r$ is estimated from test data using \citep[Algorithm 1]{aolaritei2025conformal} for the $\infty$-Wasserstein distance.  \textbf{Case-3:} This baseline uses the robust quantile $C_{\tau\mid t,i}^{\mathrm{rob}}:=\sigma_{\tau \mid t,i}\cdot C+\Delta_\mathsf{d}(r;R^{(1)}, \dots, R^{(K)})$  where $r$ is computed analytically using the bound in equation \eqref{eq:analytic:bound}. For $\mathsf{d}$, we here use the $\infty$-Wasserstein distance.

\textbf{Simulation Setup.} We use the ORCA simulator~\citep{van2008reciprocal} configured as a two-lane corridor with one ego agent and two moving obstacles \(A\) and \(B\) (see Figure \ref{fig:trajectory}). The ego starts in the right lane; obstacle \(A\) moves oppositely in the same lane, and obstacle \(B\) travels in the other lane in the same direction. The ego must safely maneuver (e.g., lane changes) to avoid collisions and reach its goal. Its dynamics follow a kinematic bicycle model discretized via forward Euler with sampling time \(\Delta=1/8\)~s~\citep{kong2015kinematic}. The ego vehicle has position $p_t:=(x_t,y_t)$, orientation $\theta_t$, and velocity $v_t$. The control inputs are the steering angle $\phi_t$ and the acceleration $a_t$.

\refstepcounter{figure}
\label{fig:trajectory}

\begin{figure}[ht]
    \centering
    \includegraphics[width=\textwidth]{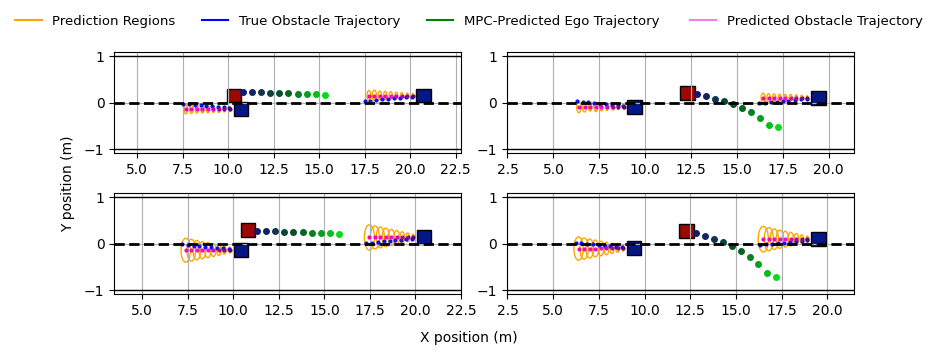}
    \par
    \justifying
    \small
    \noindent Figure \thefigure: Shown are the ego agent (red), actual (blue) and predicted (pink) trajectories of other agents (blue), along with their prediction regions (yellow). In the top row (case-0), coverage failure leads to collisions. In the bottom row (case-3), the robust prediction regions ensure safety.
\end{figure}

\textbf{Calibration Distribution.} The obstacle agents start at \(\mathbf{p}_A^{(0)}=[20,-0.5]\) and \(\mathbf{p}_B^{(0)}=[30,0.5]\). Their goals are \(\mathbf{g}_i=[x_i^{\text{goal}},y_i^{\text{goal}}]\), with \(x_A^{\text{goal}}=0\), \(x_B^{\text{goal}}=10\), and \(y_i^{\text{goal}}\sim\mathcal{U}[y_i^{\text{nom}}\!-\!0.015,\,y_i^{\text{nom}}\!+\!0.015]\), where \(y_A^{\text{nom}}=-0.5\), \(y_B^{\text{nom}}=0.5\). Velocities are sampled as \(v_i\sim\mathcal{U}[0.5,1.0]\) and adjusted by ORCA for feasible obstacle motion. This defines the \emph{training distribution}, from which \emph{test distributions} are obtained by perturbing velocities and goals as described below. 

\textbf{Introducing Distribution Shifts.} 
Each test environment \(\mathcal{E}_k\) represents a controlled shift generated by perturbing velocities \(v_i\!\sim\!\mathcal{U}[0.55,1.05]\) (common across all test sets) and varying goal latitudes while fixing \(x_A^{\text{goal}}\!=\!0\) and \(x_B^{\text{goal}}\!=\!10\). Specifically, \(\mathcal{E}_k\) is defined by \(y_A^{\text{goal}}\!\sim\!\mathcal{U}(\mathcal{Y}_A^{(k)})\) and \(y_B^{\text{goal}}\!\sim\!\mathcal{U}(\mathcal{Y}_B^{(k)})\), where \(\mathcal{Y}_A^{(k)}\!=\![-0.06\!+\!0.03(k\!-\!1),-0.03\!+\!0.03(k\!-\!1)]\) and \(\mathcal{Y}_B^{(k)}\!=\!-\mathcal{Y}_A^{(k)}\), for \(k\!=\!1,\dots,10\). Increasing \(k\) shifts the goal latitudes smoothly outward, yielding a sequence of test environments \(\{\mathcal{E}_k\}_{k=1}^{10}\) with progressively larger distributional shifts.

\textbf{Distribution Splits and Nuisance Parameters. } Out of the ten shifted environments \(\{\mathcal{E}_k\}_{k=1}^{10}\), the alternate sets \(\{\mathcal{E}_2,\mathcal{E}_4,\mathcal{E}_6,\mathcal{E}_8,\mathcal{E}_{10}\}\) are used to train the conditional diffusion model, while the remaining \(\{\mathcal{E}_1,\mathcal{E}_3,\mathcal{E}_5,\mathcal{E}_7,\mathcal{E}_9\}\) serve as unseen test distributions. Each environment is characterized by a scalar nuisance parameter \(\eta\), whose distribution varies across environments with nearly disjoint support. In our MPC setting, a short buffer interval allows estimation of \(\eta\), defined  as the mean absolute lateral velocity over the buffer: $\eta=\tfrac{1}{2T_b}\!\sum_{i\in\{A,B\}}\!\sum_{t=0}^{T_b-1}\!\left|\tfrac{y_i(t+1)-y_i(t)}{\Delta}\right|,$ where \(y_i(t)\) is the lateral position of obstacle \(i\) at time \(t\) with sampling interval \(\Delta=1/8~\mathrm{s}\). Evaluation across all \(\{\mathcal{E}_k\}_{k=1}^{10}\) environments assesses the model’s interpolation under intermediate or unseen shifts.

\textbf{Generalization of the Conditional Diffusion Model.}  We evaluate the generalization capability of our trained conditional diffusion model by measuring the empirical $2$-Wasserstein distance between synthetically generated non-conformity score samples and samples drawn from a fixed test environment. Synthetic samples are generated by varying the nuisance parameter $\eta$, while the test environment is held fixed. We denote by $\mathcal{E}_j$ a test environment characterized by nuisance parameter $\eta_j$. As shown in Figure \ref{fig:w2_vs_nuisance}, the Wasserstein discrepancy is minimized when the nuisance parameter $\eta$ used for synthetic generation matches the nuisance parameter $\eta_j$ of the corresponding test environment $\mathcal{E}_j$. This indicates that the synthetic samples most accurately approximate the test-time non-conformity score distribution when the nuisance is correctly matched. As $\eta$ deviates from $\eta_j$, the $2$-Wasserstein distance increases, reflecting a growing distributional mismatch.

\refstepcounter{figure}
\label{fig:w2_vs_nuisance}

\begin{figure}[ht]
    \centering
    \includegraphics[width=\textwidth]{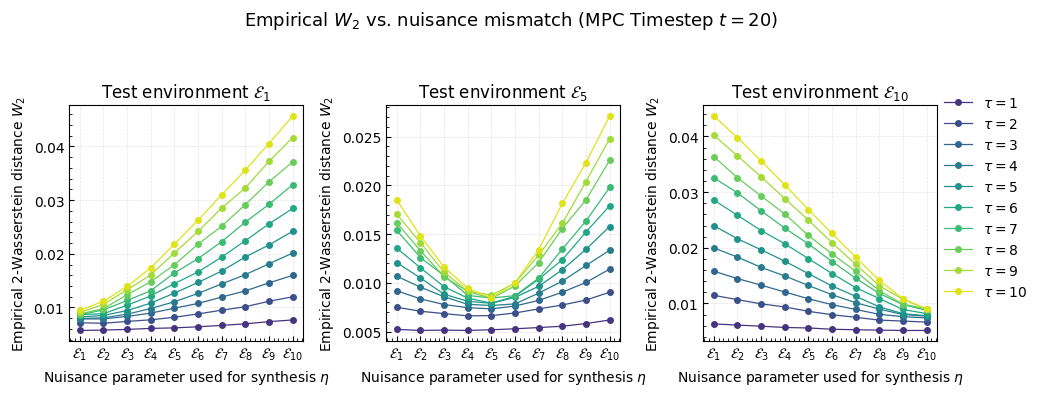}
    \par
    \justifying
    \small
    \noindent Figure \thefigure: Empirical 2-Wasserstein distance \(W_2\) between synthetically generated trajectory distributions and the true test-environment distribution as a function of the nuisance parameter $\eta$ used for synthetic generation at time $t = 20$. Each subplot corresponds to a different test environment $\mathcal{E}_j$ characterized by nuisance parameter $\eta_j$, and curves within each subplot represent different prediction horizons $\tau$. The Wasserstein distance is minimized when $\eta = \eta_j$ and increases as the nuisance parameter deviates from the test environment, indicating degraded approximation of the test-time non-conformity distribution.
\end{figure}

\textbf{Time Complexity.} At each MPC timestep, our approach incurs additional computation due to synthetic sample generation for the conditional diffusion model, beyond the computational cost of solving the MPC optimization problem. When executed without synthetic sampling, the MPC optimization alone requires approximately $42\,\mathrm{ms}$ (averaged) per timestep. Incorporating synthetic data generation increases the per-step computation time, as shown in Figure \ref{fig:time_complexity}. Despite this additional overhead, the total average computation time per MPC step remains well below a $1\,\mathrm{s}$ upper limit for practical sample sizes. In particular, for $K$ in the range of $500$--$1000$ synthetic samples—empirically sufficient to achieve a conformal miscoverage level of $\delta = 0.1$—the computation time remains under $0.5\,\mathrm{s}$. This demonstrates that the proposed method is computationally feasible for real-time MPC while providing statistically calibrated uncertainty guarantees.

\refstepcounter{figure}
\label{fig:time_complexity}

\begin{figure}[ht]
    \centering
    \includegraphics[width=0.6\textwidth]{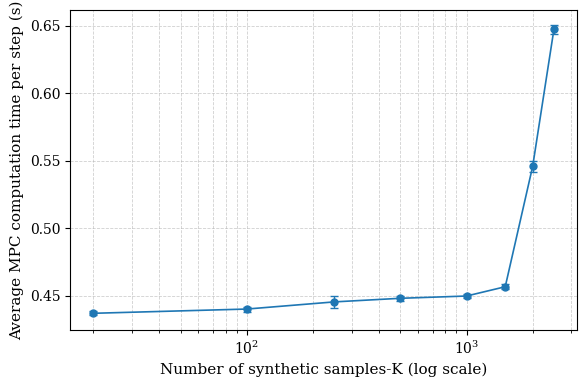}
    \par
    \justifying
    \small
    \noindent Figure \thefigure: Average computation time per MPC step versus the number of synthetic samples $K$ (log scale), including both MPC optimization and synthetic sample generation.
\end{figure}

\refstepcounter{figure}
\label{fig:safety_coverage}

\begin{figure}[ht]
    \centering
    \includegraphics[width=0.9\textwidth]{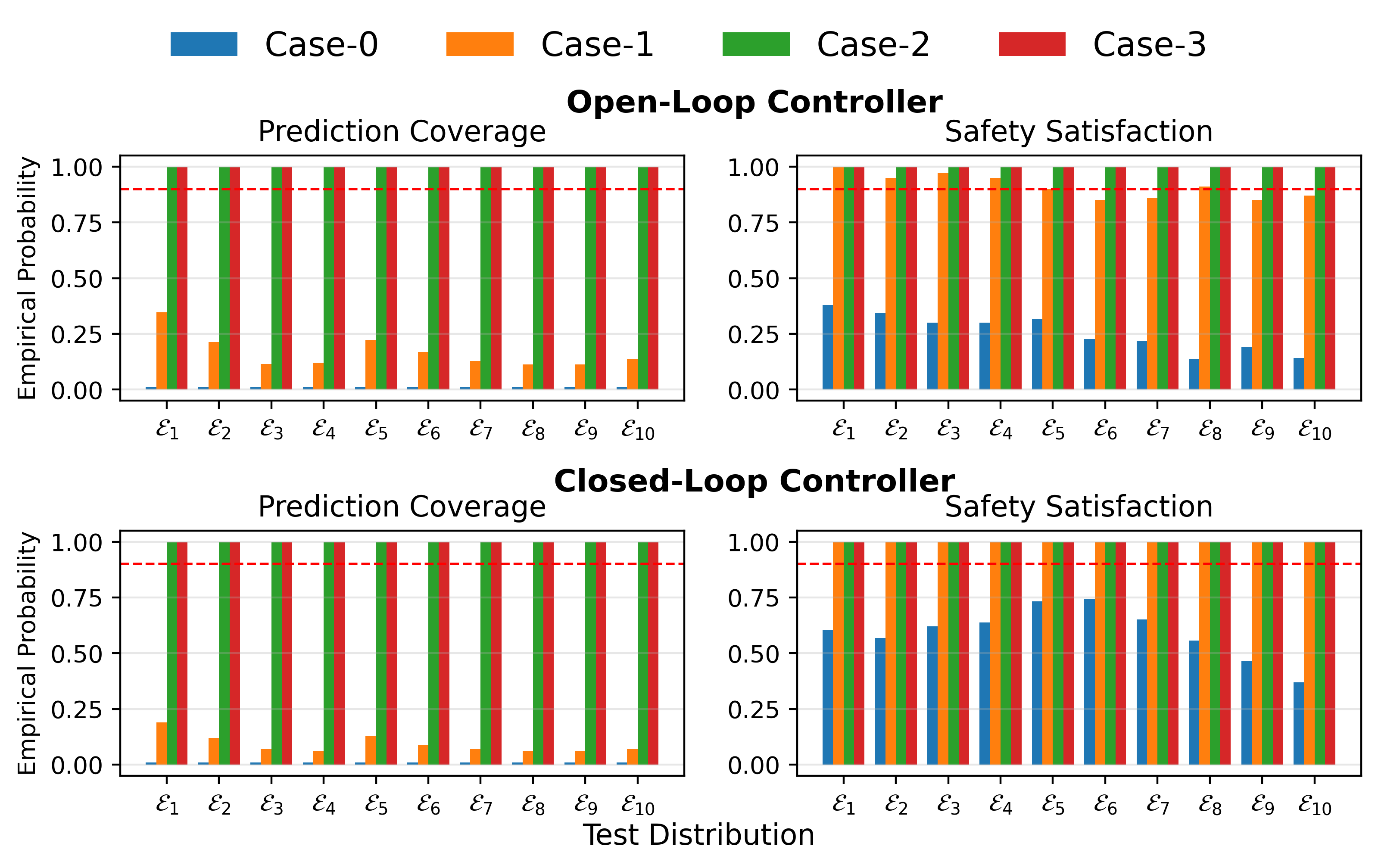}
    \par
    \justifying
    \small
    \noindent Figure \thefigure: Comparison of prediction coverages and safety across ten test environments for Cases 0-3 under open- and closed-loop controllers. The red dashed line denotes $1-\delta=0.9$.
\end{figure}

\textbf{Collision and Coverage Statistics. } Figure \ref{fig:safety_coverage} compares safety coverage $\mathbb{P}_{\mathcal{D}_{\mathrm{test}}}
      \!( c(x_t,Y_t)\ge 0,\;\forall t\in\{0,\dots,T\})$ and prediction coverage $\mathbb{P}_{\mathcal{R}_{\mathrm{test}}}
    ( \|\hat{Y}_{\tau \mid t,i}- Y_{\tau,i}||\le C_{\tau\mid t,i}^{\mathrm{rob}}, \; \forall (t,\tau,i)\in \mathcal{S}  )$  across the ten test environments  \(\{\mathcal{E}_k\}_{k=1}^{10}\) for both open-loop controller (only solving \eqref{eq:mpc} at time $t=0$)  and closed-loop controller (the MPC solving \eqref{eq:mpc} iteratively). Collisions occur only when the true obstacle trajectories fall significantly outside their predicted regions. In the open-loop setting, where no feedback correction is applied, \textbf{Case-0} shows frequent safety violations due to fixed calibration data, while \textbf{Case-1} exhibits violations despite using only synthetic data without robustification. Feedback present in the closed-loop controller mitigates these failures, reducing collisions even when coverage deteriorates. Both \textbf{Case-2} and \textbf{Case-3} maintain near-perfect safety and the target $1 - \delta = 0.9$ coverage across all environments, demonstrating that robust prediction regions effectively ensure safe operation under distributional shifts without being overly conservative, see Figure \ref{fig:trajectory}.

\section{Conclusion} We proposed a planning framework that maintains probabilistic safety guarantees for autonomous systems under distribution shifts. We achieve this by using robust conformal prediction along with data generated by conditional diffusion models which capture the natural variations in the test-time environment. By conditioning on an interpretable nuisance parameter, our model generates synthetic data that reflects the test-time environment variations, enabling safe MPC planning without (or only little) access to deployment data. Empirical evaluations in the ORCA simulator demonstrate that our planning method achieves reliable coverage and safety across both seen and unseen environments. We plan to extend this work in the future by considering distribution shifts which occur due to the response of the environment or obstacle agents to the ego behavior and planning. Another interesting direction would be to extend our setting to non-stationary environments where the environment experiences not just a single but multiple test-time distribution shifts.


\newpage
\bibliography{robustcp}

\newpage

\section{Appendix}
\label{sec:appendix}

\subsection{Robust Conformal Prediction and Computation of  $\Delta_\mathsf{d}$}
\label{app:robust_cp_details}

This section provides details on the computation and interpretation of the robustness correction term $\Delta_\mathsf{d}(r;R^{(1)}, \dots, R^{(K)})$ introduced in equation~\eqref{eq:robust_additive_form}. We distinguish between two principal formulations in the literature: the \emph{$f$-divergence ambiguity} approach~\citep{cauchois2024robust} and the \emph{L\'evy--Prokhorov (LP)} ambiguity approach~\citep{aolaritei2025conformal}.

\subsection*{(a) $f$-Divergence--Based Robust Conformal Prediction}

The $f$-divergence ambiguity formulation~\cite{cauchois2024robust} defines the admissible uncertainty set around the training score distribution as
\[
    \mathcal{U}^{f}_{r}(\mathcal{R}_{\mathrm{train}})
    = \{\, \mathcal{Q} : D_f(\mathcal{Q}\,\|\,\mathcal{R}_{\mathrm{train}}) \le r \,\},
\]
where $D_f(\cdot\|\cdot)$ denotes an $f$-divergence and $r>0$ bounds the shift magnitude.  
Let
\[
R^{(1)},\dots,R^{(K)} \overset{\text{i.i.d.}}{\sim} \mathcal{R}_{\mathrm{train}},
\qquad
R^{(0)} \sim \mathcal{R}_{\mathrm{test}}.
\]
Then, calibration scores from $\mathcal{R}_{\mathrm{train}}$ can be used to construct a robust quantile that guarantees coverage under all $\mathcal{R}_{\mathrm{test}}\!\in\!\mathcal{U}^{f}_{r}$.

\begin{lemma}[Adapted from~\cite{cauchois2024robust}] \label{lemma:fdiv_app}
Let $R^{(0)},\dots,R^{(K)}$ be independent random variables with $R^{(0)}\!\sim\!\mathcal{R}_{\mathrm{test}}$ and $R^{(1)},\dots,R^{(K)}\!\sim\!\mathcal{R}_{\mathrm{train}}$, satisfying $\mathcal{R}_{\mathrm{test}}\in \mathcal{U}^{f}_{r}(\mathcal{R}_{\mathrm{train}})$.  
For a failure probability $\delta\!\in\!(0,1)$, the guarantee $\mathbb{P}\big(R^{(0)} \le C^{\mathrm{rob}}\big) \;\ge\; 1-\delta$ holds for 
\[
    C^{\mathrm{rob}}
    := \mathrm{Quantile}_{1-\tilde{\delta}}\!\big(R^{(1)},\dots,R^{(K)}\big),
\]
where the adjusted level $\tilde{\delta}$ is obtained via one-dimensional convex programs:
\[
\begin{aligned}
    \delta_K
    &:= 1 - g\!\Big((1+\tfrac{1}{K})\,g^{-1}(1-\delta)\Big),\\
    \tilde{\delta}
    &:= 1 - g^{-1}(1-\delta_K),
\end{aligned}
\]
where the auxiliary function $g:[0,1]\!\to\![0,1]$ is defined as
\[
    g(\beta)
    := \inf\Big\{
        z\in[0,1]
        \;\Big|\;
        \beta f\!\Big(\frac{z}{\beta}\Big)
        + (1-\beta)f\!\Big(\frac{1-z}{1-\beta}\Big)
        \le r
    \Big\},
\]
and its pseudo-inverse $g^{-1}(\tau) := \sup\{\beta\in[0,1]\mid g(\beta)\le\tau\}$.
\end{lemma}

\noindent\textbf{Note.}
 Both $g$ and $g^{-1}$ are convex programs in one variable and can be efficiently computed (e.g., via line search).  
In certain divergence choices, closed-form solutions exist; for example, with the total-variation generator $f(r)=\tfrac{1}{2}|r-1|$, one obtains $g(\beta)=\max(0,\beta-r)$, which yields an explicit robustness adjustment.

 Given Lemma \ref{lemma:fdiv_app}, we can easily compute $\Delta_\mathsf{d}(r;R^{(1)}, \dots, R^{(K)})$  for the $f$-divergence as
\[
    \Delta_\mathsf{d}(r;R^{(1)}, \dots, R^{(K)})
    =  C^{\mathrm{rob}} - C.
\]

\subsection*{(b) Lévy--Prokhorov--Based Robust Conformal Prediction}

The Lévy--Prokhorov (LP) ambiguity set models distribution shifts by allowing both local 
(Wasserstein-type) and global (total-variation-type) perturbations and is defined as
\begin{equation}
    \mathcal{U}_{\varepsilon,\rho}^{\mathsf{LP}}(\mathcal{R}_{\mathrm{train}})
    = \{\, \mathcal{Q} : \mathsf{LP}_{\varepsilon}(\mathcal{R}_{\mathrm{train}},\mathcal{Q}) \le \rho \,\},
\end{equation}
where the LP pseudo-metric is defined as
\[
    \mathsf{LP}_{\varepsilon}(\mathcal{R}_{\mathrm{train}},\mathcal{Q})
    = \inf_{\gamma \in \Gamma(\mathcal{R}_{\mathrm{train}},\mathcal{Q})}
      \int \mathbf{1}\{\|z_1 - z_2\| > \varepsilon\} \, d\gamma(z_1,z_2).
\]
Here, $\Gamma(\mathcal{R}_{\mathrm{train}},\mathcal{Q})$ denotes the set of all couplings with marginals $\mathcal{R}_{\mathrm{train}}$ and $\mathcal{Q}$. 
As shown in~\cite{aolaritei2025conformal}, this ambiguity set can be decomposed as
\[
    \mathcal{U}_{\varepsilon,\rho}^{\mathsf{LP}}(\mathcal{R}_{\mathrm{train}})
    = \bigcup_{\widetilde{\mathcal{P}} : W_\infty(\mathcal{R}_{\mathrm{train}}, \widetilde{\mathcal{P}}) \le \varepsilon}
      \{\, \mathcal{Q} : \mathrm{TV}(\widetilde{\mathcal{P}}, \mathcal{Q}) \le \rho \,\},
\]
highlighting that $\varepsilon$ controls local geometric shifts (through $W_\infty$) while $\rho$ allows 
a $\rho$-fraction of global mass displacement. Similar to the $f$-divergence, the calibration scores from $\mathcal{R}_{\mathrm{train}}$ can now be  used to construct a robust quantile that guarantees coverage under all $\mathcal{R}_{\mathrm{test}}\!\in\!\mathcal{U}^{\mathsf{LP}}_{r}$.

\begin{lemma}[Corollary 4.2 from \cite{aolaritei2025conformal}] \label{lemma:LP_dist}
Let $R^{(0)},\dots,R^{(K)}$ be independent random variables with $R^{(0)}\!\sim\!\mathcal{R}_{\mathrm{test}}$ and $R^{(1)},\dots,R^{(K)}\!\sim\!\mathcal{R}_{\mathrm{train}}$, satisfying $\mathcal{R}_{\mathrm{test}}\in \mathcal{U}_{\varepsilon,\rho}^{\mathsf{LP}}(\mathcal{R}_{\mathrm{train}})$.
For a failure probability $\delta\!\in\!(0,1)$, define the  set
\[
    C^{\mathrm{rob}}
    :=\mathrm{Quantile}_{1-\beta+\rho}\!\big(R^{(1)},\dots,R^{(K)}\big)  + \varepsilon 
\]
with $\beta
    =
    \delta + \frac{\delta - \rho - 2}{K}$. Then, $\mathbb{P}\big(R^{(0)} \le C^{\mathrm{rob}}\big) \;\ge\; 1-\delta$.
    

\end{lemma}

Looking at Lemma \ref{lemma:LP_dist}, we note that the robust quantile $C^{\mathrm{rob}}$ is obtained by  inflating the standard quantile in two ways:
(i) $\rho$ shifts the quantile \emph{level} to $1-\beta+\rho$, and  
(ii) $\varepsilon$ introduces a uniform additive displacement reflecting the largest possible change
in score induced by a local perturbation.

\subsection{Generative Modeling and Conditional Diffusion Models} 
\label{app:dif_model}
This section provides a concise overview of the diffusion modeling framework used in our approach.
Specifically, we
(a) briefly review the formulation of Denoising Diffusion Probabilistic Models (DDPMs),
(b) introduce conditional diffusion models with context variables, and
(c) describe the CARD architecture and its associated training objective.
We follow standard DDPM constructions from~\cite{ho2020denoising,song2020score}
and their conditional extension developed in~\cite{han2022card}.

\paragraph*{(a) Diffusion Models.} Diffusion models define a generative process by progressively perturbing data samples \(s_0 \sim p_{\text{data}}\) -- drawn from an unknown underlying data distribution $p_{\text{data}}$ -- with Gaussian noise, and then learning to reverse this process to recover clean samples.
The forward process constructs a sequence of latent variables \(\{s_{j}\}_{j=0}^{T_{\text{diff}}}\) via a fixed variance schedule \(\{\beta_{j}\}_{j=1}^{T_{\text{diff}}}\)
\begin{equation}
    q(s_{j} \mid s_{j-1})
    = \mathcal{N}\!\left(
        \sqrt{1 - \beta_{j}}\,s_{j-1},\, 
        \beta_{j}\mathbf{I}
    \right),
\end{equation}
which admits a closed-form expression for any diffusion step
\begin{equation}
    q(s_{j} \mid s_0)
    = \mathcal{N}\!\left(
        \sqrt{\bar{\alpha}_{j}}\,s_0,\,
        (1 - \bar{\alpha}_{j})\mathbf{I}
    \right)
\end{equation}
where $\bar{\alpha}_{j} = \prod_{s=1}^{j} (1 - \beta_s)$. The reverse (denoising) process approximates the true posterior \(q(s_{j-1} \mid s_{j})\) by a parameterized Gaussian:
\begin{equation}\label{eq:rev_diff}
    p_\theta(s_{j-1}\mid s_{j})
    =
    \mathcal{N}\!\left(
        \mu_\theta(s_{j}, j),
        \Sigma_{j}
    \right),
\end{equation}
where the mean is reparameterized using a neural network \(\epsilon_\theta(s_{j}, j)\) that predicts the added noise:
\begin{equation}
    \mu_\theta(s_{j}, j)
    =
    \frac{1}{\sqrt{\alpha_{j}}}
    \left(
        s_{j}
        - \frac{\beta_{j}}{\sqrt{1 - \bar{\alpha}_{j}}}\,
          \epsilon_\theta(s_{j}, j)
    \right).
\end{equation}
The reverse-process covariance is typically chosen as $\Sigma_j = \beta_j \mathbf{I}$. This construction defines a Markov chain that gradually transforms Gaussian noise into samples that are drawn from the learned data distribution \(p_\theta\) and approximate \(s_0 \sim p_{\text{data}}\). The process of generating new data then consists  of sampling $x_{T_{\text{diff}}}\sim \mathcal{N}(0,\mathbf{I})$ and recursively applying \eqref{eq:rev_diff}.

\paragraph*{(b) Conditional Diffusion Models (CARD design).}

To model structured dependencies or contextual information, diffusion models can be extended to represent a family of conditional distributions \(p_\theta(s_0 \mid \mathbf{c})\), where \(\mathbf{c}\) denotes conditioning covariates such as control inputs, environmental descriptors, or temporal indices. In the CARD framework~\citep{han2022card}, the key distinction is that the \emph{forward} diffusion becomes context-dependent via a pre-trained conditional mean encoder \(f_\phi(\mathbf{c})\). Specifically, the forward process is modified to drift toward the context-dependent mean:
\begin{equation}
    q(s_j \mid s_{j-1}, \mathbf{c})
    =
    \mathcal{N}\!\Big(
        \sqrt{1-\beta_j}\,s_{j-1}
        + \big(1-\sqrt{1-\beta_j}\big) f_\phi(\mathbf{c}),\;
        \beta_j \mathbf{I}
    \Big).
\end{equation}
This modification ensures that the noisy variables remain centered around a context-informed trajectory, so that the reverse process learns deviations around a meaningful conditional mean rather than reconstructing it from scratch.

The reverse process then follows the DDPM parameterization but conditioned on \(\mathbf{c}\):
\begin{equation}
    p_\theta(s_{j-1}\mid s_j,\mathbf{c})
    =
    \mathcal{N}\!\big(
        \mu_\theta(s_j,j,f_\phi(\mathbf{c})),\;
        \Sigma_j
    \big),
\end{equation}
where the mean is expressed as
\begin{equation}
    \mu_\theta(s_j,j,f_\phi(\mathbf{c}))
    =
    \frac{1}{\sqrt{\alpha_j}}
    \left(
        s_j
        -
        \frac{\beta_j}{\sqrt{1-\bar{\alpha}_j}}\,
        \epsilon_\theta(s_j,j,f_\phi(\mathbf{c}))
    \right).
\end{equation}
The learnable parameters \(\theta\) inhabit the noise-prediction network \(\epsilon_\theta\), whereas \(f_\phi(\mathbf{c})\) is provided by a pre-trained encoder. The reverse-process covariance $\Sigma_j$ has the same definition as before, typically chosen as $\Sigma_j = \beta_j \mathbf{I}$.

\paragraph*{(c) Training Objective.}
The noise-prediction network $\epsilon_\theta$ is trained with a conditional DDPM loss adapted to the CARD forward process:
\[
\mathcal{L}_{\mathrm{diff}}(\theta)
=
\mathbb{E}_{s_0,\mathbf{c},j,\varepsilon}
\Big[
    \big\|\varepsilon - \epsilon_\theta\big(s_j,j,f_\phi(\mathbf{c})\big)\big\|_2^2
\Big],
\]
where $s_j$ is sampled from the CARD forward kernel
\[
s_j
= \sqrt{\bar\alpha_j}\,s_0
  + (1-\sqrt{\bar\alpha_j})\,f_\phi(\mathbf{c})
  + \sqrt{1-\bar\alpha_j}\,\varepsilon,\qquad
\varepsilon\sim\mathcal{N}(0,I).
\]
The encoder $f_\phi$ is typically \emph{pre-trained} (and then fixed during diffusion training) using a simple regression objective
\[
\mathcal{L}_{\mathrm{enc}}(\phi)
=
\mathbb{E}_{s_0,\mathbf{c}}\big[\|s_0 - f_\phi(\mathbf{c})\|_2^2\big],
\]
so that $f_\phi(\mathbf{c})\approx\mathbb{E}[s_0\mid\mathbf{c}]$, using a (possibly disjoint) subset of the same data distribution employed to train the CARD model.

When the conditioning variable is omitted, the formulation reduces to the standard, unconditional diffusion model. Although the model is trained via noise regression, this objective is algebraically equivalent to denoising score matching for the DDPM/CARD forward process; the precise conversion used for the analytical bound is provided in Appendix~\ref{app:analytical_bound}.

\subsection{Analytical Error Bound for CARD}
\label{app:analytical_bound}
For a given context \(\mathbf{c}\), \cite{li2025non} obtained a computable upper bound on the 2-Wasserstein distance
\[
r = W_2\!\big(q(\cdot\mid \mathbf{c}),\,p_\theta(\cdot\mid \mathbf{c})\big),
\]
where \(q(\cdot\mid \mathbf{c})\) denotes the unknown conditional data distribution and \(p_\theta(\cdot\mid \mathbf{c})\) is the conditional distribution learned by the CARD model.

\medskip

\noindent\textbf{Strategy and theoretical foundation.} We rely on a result from~\cite{li2025non}, which upper bounds the 2-Wasserstein distance between the two aforementioned distributions in terms of the score-matching error and certain regularity constants. To apply this result to the CARD model used in our framework, we proceed in three steps: (a) we relate the noise-matching training objective used in CARD to the score-matching objective that appears in the result from~\cite{li2025non}; (b) we investigate the required regularity assumptions and explicitly compute the associated constants for the CARD dynamics; and (c) we specialize the resulting Wasserstein bound to our conditional setting and present it as a theorem tailored to our model.

\paragraph*{(a) Equivalence of noise and score objectives.}

CARD trains a denoising network by regressing the additive noise \(\varepsilon\) (the
\(\varepsilon\)-prediction or ``noise'' objective), whereas the analytical error bound in~\cite{li2025non} is formulated in terms of a score-matching objective. For DDPM-style forward processes, these two objectives are  equivalent. We make this equivalence explicit below in order to connect the CARD training error to the quantity appearing in the Wasserstein error bound from ~\cite{li2025non}.

\textbf{Forward marginal and notation.} For both DDPM and CARD, the forward diffusion marginal at diffusion step \(j\) is
Gaussian. To simplify notation, we introduce the shorthand\footnote{Note that we introduce  the expression $\kappa$ for concise writing and has no other significance in related literature.}
\[
\kappa_j(s_0,f_\phi(\mathbf{c}))=\sqrt{\bar\alpha_j}\,s_0+(1-\sqrt{\bar\alpha_j})f_\phi(\mathbf c)\
\]
which represents the context-dependent mean of the forward process. With this notation, the conditional density of \(s_j\) given \((s_0,\mathbf{c})\) is
\[
q(s_j \mid s_0,\mathbf{c})
=
\mathcal{N}\!\bigl(
\kappa_j(s_0,f_\phi(\mathbf{c})),\,
(1-\bar\alpha_j)\mathbf I
\bigr),
\]
that is, the value of the Gaussian probability density with mean \(\kappa_j(s_0,\mathbf c)\) and covariance \((1-\bar\alpha_j)\mathbf I\).

\textbf{Score of the forward marginal.} The \emph{score} of a distribution is defined as the gradient of its log-density with respect to the variable of interest. For the Gaussian density above, this is
\[
\log q(s_j \mid s_0,\mathbf{c})
=
-\frac{1}{2(1-\bar\alpha_j)}
\|s_j-\kappa_j(s_0,f_\phi(\mathbf{c}))\|^2
+\text{const}.
\]
Differentiating with respect to \(s_j\) yields the \emph{score} as expressed below:
\[
\nabla_{s_j}\log q(s_j \mid s_0,\mathbf{c})
=
-\frac{s_j-\kappa_j(s_0,f_\phi(\mathbf{c}))}{1-\bar\alpha_j}.
\]

\textbf{Sampling interpretation (reparameterization).} The expression for $q(s_j \mid s_0,\mathbf{c})$ defined previously describes the distribution of \(s_j\). Equivalently, a random sample \(s_j\) drawn from this Gaussian distribution can be written using the standard reparameterization of a normal random variable as
\[
s_j
=
\kappa_j(s_0,f_\phi(\mathbf{c}))
+
\sqrt{1-\bar\alpha_j}\,\varepsilon,
\qquad
\varepsilon\sim\mathcal{N}(0,\mathbf I).
\]
This representation does not introduce a new assumption; it is simply an explicit way of expressing samples from
\(\mathcal{N}(\kappa_j,(1-\bar\alpha_j)\mathbf I)\). Substituting this expression for \(s_j\) into the score yields the key noise–score identity:
\[
\nabla_{s_j}\log q(s_j\mid s_0,\mathbf{c})
= -\frac{1}{\sqrt{1-\bar\alpha_j}}\,\varepsilon.
\]

\textbf{From noise regression to score matching.} As previously mentioned, CARD is trained by minimizing the noise regression error
\[
\|\varepsilon_\theta(s_j,j,f_\phi(\mathbf{c}))-\varepsilon\|^2.
\]
Using the noise-score identity above, the noise variable can be written in terms of the score as
\[
\varepsilon
=
-\sqrt{1-\bar\alpha_j}\;
\nabla_{s_j}\log q(s_j\mid s_0,\mathbf{c}).
\]
Substituting this expression into the noise regression error yields
\[
\|\varepsilon_\theta(s_j,j,f_\phi(\mathbf{c})) - \varepsilon\|^2
=
\|\varepsilon_\theta(s_j,j,f_\phi(\mathbf{c}))
+
\sqrt{1-\bar\alpha_j}\;
\nabla_{s_j}\log q(s_j\mid s_0,\mathbf{c})\|^2.
\]
We now define the score estimator implicitly induced by the noise-prediction network as
\[
s_\theta(s_j,j,f_\phi(\mathbf{c}))
=
-\frac{1}{\sqrt{1-\bar\alpha_j}}\,
\varepsilon_\theta(s_j,j,f_\phi(\mathbf{c})).
\]
If the model predicts the score \(s_\theta(s_j,j,f_\phi(\mathbf{c}))\) or the noise \(\varepsilon_\theta(s_j,j,f_\phi(\mathbf{c}))\), the identity above implies the pointwise relation. With this definition, the error term becomes
\[
\|\varepsilon_\theta(s_j,j,f_\phi(\mathbf{c})) - \varepsilon\|^2
=
\|-\sqrt{1-\bar\alpha_j}\,
\bigl(
s_\theta(s_j,j,f_\phi(\mathbf{c})) - \nabla_{s_j}\log q(s_j\mid s_0,\mathbf{c})
\bigr)\|^2.
\]
Taking squared expectations on both sides gives
\[
\mathbb{E}\!\left[\|\varepsilon_\theta(s_j,j,f_\phi(\mathbf{c}))-\varepsilon\|^2\right]
=
(1-\bar\alpha_j)\,
\mathbb{E}\!\left[
\|s_\theta(s_j,j,f_\phi(\mathbf{c}))-\nabla_{s_j}\log q(s_j\mid s_0,\mathbf{c})\|^2
\right].
\]
Denoting the noise mean-squared error by \(E(j)=\mathbb{E}[\|\varepsilon_\theta(s_j,j,f_\phi(\mathbf{c}))-\varepsilon\|^2]\) and the score mean-squared error by \(H(j)=\mathbb{E}[\|s_\theta(s_j,j,\mathbf{c})-\nabla\log q(s_j\mid s_0,f_\phi(\mathbf{c}))\|^2]\), we obtain the exact relation
\[
H(j)
=
\frac{E(j)}{1-\bar\alpha_j}.
\]
Thus, the noise-regression objective used to train CARD is exactly equivalent, up to a known scaling factor, to the score-matching error required by the Wasserstein error bound from \cite{li2025non}. Therefore, inserting \(H_j\) into \cite[Theorem 1]{li2025non} results directly in Theorem~\ref{thm:5} (stated below) by translating  the score-regression quantity  into a noise-regression empirical error.

\paragraph*{(b) Assumptions and regularity conditions.} The analytical bound of~\cite{li2025non} requires mild regularity conditions on the diffusion dynamics and the associated score functions. In particular, the bound depends on two step-dependent Lipschitz constants: \(L_1(j)\) (associated with the forward diffusion drift) and \(L_2(j)\) (associated with the score function) where \(j \in \{1,\dots,T_{\sf diff}\}\) denotes the diffusion step for a total of $T_{\sf diff}$ diffusion steps. The dependence on the diffusion step \(j\) arises naturally because both the drift and the score vary across noise levels in the diffusion process. Specifically, the assumptions are:
(i) a one-sided Lipschitz condition on the forward drift \(b_j(s) = -\tfrac{1}{2}\beta_j (s - f_\phi(\mathbf{c}))\) with constant \(L_1(j)\); (ii) a one-sided Lipschitz condition on both the true and learned score functions with constant \(L_2(j)\);
(iii) a bounded diffusion variance schedule \(\{\beta_j\}\); and (iv) consistency of the pre-trained conditional mean encoder, namely that for any \(\delta>0\) there exists \(a_0\) such that for all training epochs
\(a>a_0\), $\mathbb{P}\!\left(
\|f_\phi(\mathbf{c})-\mathbb{E}[s_0\mid \mathbf{c}]\|_2<\delta
\right)\to 1$. These conditions match the regularity assumptions required
in~\cite{li2025non}.

\textbf{Analytical form of the drift Lipschitz constant \(L_1(j)\).}
The analytical bound of~\cite{li2025non} assumes that the forward drift satisfies a Lipschitz-type regularity condition. In the CARD formulation, the forward diffusion drift at diffusion step \(j\) is
given by $b_j(s) = -\tfrac{1}{2}\beta_j \big(s - f_\phi(\mathbf{c})\big)$. For any two states \(s_a,s_b \in \mathbb{R}^d\), we have $\|b_j(s_a)-b_j(s_b)\|^2
=
\left\|
-\tfrac{1}{2}\beta_j (s_a-s_b)
\right\|^2
=
\frac{\beta_j^2}{4}\,\|s_a-s_b\|^2$.
Thus, the drift satisfies the squared-norm Lipschitz inequality $\|b_j(s_a)-b_j(s_b)\|^2
\;\le\;
L_1(j)\,\|s_a-s_b\|^2$,
with the explicit analytical constant $L_1(j)=\frac{\beta_j^2}{4}$ which we use in our experiments. This verifies the required regularity condition for the forward diffusion drift used in the Wasserstein bound.

\textbf{Estimation of the score Lipschitz constant \(L_2(j)\).}
Unlike the drift, the score function is represented by a neural network $s_\theta$ and is therefore not available in closed form. The constant \(L_2(j)\) quantifies the one-sided Lipschitz continuity of the learned score function $s_\theta(s_j,j,f_\phi(\mathbf{c}))$  and is defined through the inequality
\[
\bigl(s_\theta(s_a,j,f_\phi(\mathbf{c}))-s_\theta(s_b,j,f_\phi(\mathbf{c}))\bigr)^\top (s_a-s_b)
\le
L_2(j)\,\|s_a-s_b\|^2,
\]
for all noisy states \(s_a,s_b\) at diffusion step \(j\) under context
\(\mathbf{c}\).
In our experiments, we estimate \(L_2(j)\) empirically. Therefore, we sample pairs of noisy states \((s_a,s_b)\) are sampled from the forward marginal \(q(s_j\mid s_0, \mathbf{c})\), evaluate the corresponding score network outputs, and compute the ratio
\[
\frac{
\bigl(s_\theta(s_a,j,f_\phi(\mathbf{c}))-s_\theta(s_b,j,f_\phi(\mathbf{c}))\bigr)^\top (s_a-s_b)
}{
\|s_a-s_b\|^2
}.
\]
Subsequently, we  choose the maximum value of this quantity over all the random pairs to obtain an estimate of \(L_2(j)\).

\paragraph*{(c) The Wasserstein error bound.} With the previous definitions and explanations, we are now ready to state \cite[Theorem 1]{li2025non}, which we have used in the main part of the paper.

\begin{theorem}[Analytical 2-Wasserstein Bound,~\cite{li2025non}]\label{thm:5}
Let \(q(\cdot\mid \mathbf{c})\) denote the true conditional data distribution and \(p_\theta(\cdot\mid \mathbf{c})\) the conditional distribution induced by a diffusion-based generative model trained via denoising or score matching. Suppose the drift and score functions satisfy one-sided Lipschitz and smoothness conditions with constants \(L_1(j)\) and \(L_2(j)\).
Then, the 2-Wasserstein distance between the true and the learned conditional distribution satisfies

\begin{align}
W_2\!\big(q(\cdot\mid \mathbf{c}),\,p_\theta(\cdot\mid \mathbf{c})\big)
\;\le\;
\underbrace{\sum_{j=1}^{T_{\sf diff}}\! \beta_j\, M(j)\, \sqrt{H(j)}}_{\text{score approximation error}}
\;+\;
\underbrace{
\Bigg(
\sum_{j=1}^{T_{\sf diff}}
\beta_j^2
\exp\!\Big(\sum_{s\le j} (L_1(s)+L_2(s)\beta_s)\Big)
\Bigg)^{1/2}
}_{\text{residual discretization term}}, \label{eq:full-bound}
\end{align}
where \(H(j)=E(j) / (1-\bar\alpha_j)\) and $E(j)=\mathbb{E}
\big[\|\varepsilon_\theta(s_j,j,f_\phi(\mathbf{c}))-\varepsilon\|^2\big]$ is the per-step noise MSE, 
\(M(j) = \exp\!\Big(\sum_{s\le j}\! (L_1(s)+L_2(s)\beta_s)\Big)\), 
and \(\beta_t\) are the diffusion noise coefficients. 
\end{theorem}

We note that under the regularity assumptions of~\cite{li2025non}, the residual discretization term vanishes as \(T_{\sf diff}\to\infty\) or for sufficiently small \(\beta_t\), reducing the bound to the dominant first term which we stated in the main part of the paper. Finally, to map Theorem \ref{thm:5} to our setting in the main part of the paper, we can simply view  $\bar{\mathcal{R}}_{\mathrm{test}}(\zeta)$ as $q(\cdot\mid \mathbf{c})$.

\subsection{Analytical Robustification of the Conformal Threshold}
\label{app:analytical_bound_CP}

This section explains how the Wasserstein error bound from Theorem \ref{thm:5} for the conditional diffusion model can be used to obtain $C_{\tau\mid t,i}^{\mathrm{rob}}$ that makes Lemma \ref{cor2} valid. The argument proceeds in three steps: (a) we recall how to obtain Wasserstein bounds of Lipschitz continuous functions, 
(b) we relate Wasserstein bounds from the model $p_\theta$ to Wasserstein bounds of the non-conformity scores,
and (c) we use the resulting bound to construct the constant $C_{\tau\mid t,i}^{\mathrm{rob}}$ that makes Lemma \ref{cor2} valid.

\paragraph{(a) Wasserstein bounds of Lipschitz continuous functions.}
Let $\mathcal{D}$ and $\mathcal{D}_0$ denote two probability distributions, where $\mathcal{D}$ could represent the true underlying distribution  while $\mathcal{D}_0$ could represent the learned distribution, approximating $\mathcal{D}$. Let $f$ be an $L$-Lipschitz continuous function with respect to the underlying metrics that, if applied to samples from $\mathcal{D}$ and $\mathcal{D}_0$, transforms the distributions $\mathcal{D}$ and $\mathcal{D}_0$ into the pushforward distributions $\mathcal{R}$ and $\mathcal{R}_0$, respectively. Existing bounds for $L$-Lipschitz continuous functions (see, e.g.,~\cite{villani2008optimal}) then guarantee that 
\[
W_2(\mathcal{R}, \mathcal{R}_0)
\;\le\;
L\, W_2(\mathcal{D}, \mathcal{D}_0).
\]

\paragraph{(b) Propagation of Wasserstein bound to the nonconformity score distribution.} Recall now that the original non-conformity score in Section \ref{sec:4_1}  was defined as
\begin{align}\label{eq:nonconformity_CL_appendix}
    R^{(k)}=\max_{(t,\tau,i)\in\mathcal{S}} 
    \frac{\|\hat{Y}_{\tau \mid t,i}^{(k)}- Y_{\tau,i}^{(k)}\|}{\sigma_{\tau \mid t,i}}.
\end{align}
In Section \ref{sec:4_2} we learned a conditional diffusion model $p_\theta(\bar{R} \mid \mathbf{c})$ that predicts  
$\|\hat{Y}_{\tau \mid t,i}^{(k)}- Y_{\tau,i}^{(k)}||$ for trajectories  $Y^{(k)}$. This way, we approximate the nonconformity score in  \eqref{eq:nonconformity_CL_appendix} as 
\begin{align}\label{eq:appendix_nonconformity_42}
    R^{(k)}=\max_{(t,\tau,i)\in\mathcal{S}} \frac{\bar{R}_{\tau|t,i}^{(k)}}{\sigma_{\tau \mid t,i}}
\end{align}
where $\bar{R}_{\tau|t,i}^{(k)}$ is generated by $p_\theta(\bar{R} \mid \mathbf{c})$.

Suppose next that the distribution of the learned conditional generative model $p_\theta(\bar{R} \mid \mathbf{c})$ corresponding to $ 
\bar{R}_{\tau|t,i}^{(k)}$ satisfies a Wasserstein bound of \(\varepsilon_W\) with the distribution  $\bar{\mathcal{R}}_{\mathrm{test}}(\zeta)$ corresponding to $
\|\hat{Y}_{\tau\mid t,i}^{(k)}-Y_{\tau,i}^{(k)} \|$, e.g., obtained using Theorem \ref{thm:5}. This means that
\[
W_2\left(\bar{\mathcal{R}}_{\mathrm{test}}(\zeta),p_\theta(\bar{R} \mid \mathbf{c})
\right)
\le \varepsilon_W.
\]

Next note that the function $f(s)=s/\sigma_{\tau\mid t,i}$ is $1/\sigma_{\tau\mid t,i}$-Lipschitz continuous. Similarly, the function $f(s)=\max_{(t,\tau,i)\in\mathcal{S}} s/\sigma_{\tau\mid t,i}$ is $1/\sigma_{\min}$-Lipschitz continuous where $\sigma_{\min}:=\min_{(t,\tau,i)\in\mathcal{S}}\sigma_{\tau\mid t,i}$. Finally, let $\bar{\mathcal{R}}_{\mathrm{test}}(\zeta)$ and  $p_\theta(\bar{R} \mid \mathbf{c})$ correspond to $\mathcal{D}$ and $\mathcal{D}_0$, respectively, while $\mathcal{R}$ and $\mathcal{R}_0$ follow equations \eqref{eq:nonconformity_CL_appendix} and \eqref{eq:appendix_nonconformity_42}  so that they are induced by $\mathcal{D}$ and $ \mathcal{D}_0$, respectively. Then, we can use the previous Wasserstein bounds for Lipschitz continuous functions and conclude that 
\[
W_2\left(\mathcal{R},\mathcal{R}_0
\right)
\le \varepsilon_W/\sigma_{\min}.
\]

By this result it follows that any Wasserstein error of  $\varepsilon_W$ in the data distribution induces at most a proportional Wasserstein error of $\varepsilon_W/\sigma_{\min}$ in the resulting non-conformity score distribution. This result provides a principled mechanism to quantify how uncertainty in the data-generating process propagates through the conformal mapping.

\paragraph{(c) Robust conformal threshold.} Let us now interpret $\mathcal{R}_{\mathrm{test}}$ and $ \mathcal{R}_{\mathrm{train}}$ from Lemma \ref{cor2} as $\mathcal{R}$ and $ \mathcal{R}_0$, respectively. In this case, we know that $r$ in Lemma \ref{cor2} is $\varepsilon_W/\sigma_{\min}$ so that the choice of  
\begin{equation}
C_{\tau\mid t,i}^{\mathrm{rob}}
:=
\sigma_{\tau\mid t,i}
\Big(
C
+
\Delta_{\mathsf d}\!\left(\tfrac{\varepsilon_W}{\sigma_{\min}};\,
R^{(1)},\dots,R^{(K)}\right)
\Big)
\end{equation}
makes Lemma \ref{cor2} valid.

\end{document}